\documentclass[a4paper,twocolumn,11pt,accepted=2021-05-20]{quantumarticle}
\pdfoutput=1
\usepackage[utf8]{inputenc}
\usepackage[english]{babel}
\usepackage[T1]{fontenc}
\usepackage{amsmath}
\usepackage[numbers,sort&compress]{natbib}
\usepackage{hyperref}
\usepackage{tikz}
\usepackage{lipsum}

\usepackage{graphicx}
\usepackage{amssymb}
\usepackage{amsfonts}
\usepackage{amsthm}
\usepackage{color}
\usepackage{amsfonts}
\usepackage{xcolor}
\usepackage{bbold}
\usepackage{bm} 
\usepackage{bbold}
\usepackage{bm}        
\usepackage[mathscr]{eucal}
\usepackage[normalem]{ulem}
\usepackage[all]{hypcap}
\usepackage{multirow}
\usepackage{nicefrac}
\usepackage[bottom]{footmisc}
\usepackage{url}

\hypersetup{colorlinks=true, linkcolor=blue, citecolor=blue, urlcolor=blue}

\newtheorem{lemma}{\bf{Lemma}}
\newtheorem{theorem}{Theorem}

\begin{document}

\title{Experimental localisation of quantum entanglement through monitored classical mediator}

\author{Soham Pal} \email{soham.pal@students.iiserpune.ac.in}
\affiliation{Department of Physics, Indian Institute of Science Education and Research, Pune 411008, India}

\author{Priya Batra}
\affiliation{Department of Physics, 
	     Indian Institute of Science Education and Research, Pune 411008, India}

\author{Tanjung Krisnanda}
\affiliation{School of Physical and Mathematical Sciences, Nanyang Technological University, Singapore 637371, Singapore}

\author{Tomasz Paterek}
\affiliation{School of Physical and Mathematical Sciences, Nanyang Technological University, Singapore 637371, Singapore}
\affiliation{MajuLab, International Joint Research Unit UMI 3654, CNRS, Universit\'{e} C\^{o}te d'Azur, Sorbonne Universit\'{e}, National University of Singapore, Nanyang Technological University, Singapore}
\affiliation{Institute of Theoretical Physics and Astrophysics, Faculty of Mathematics, Physics and Informatics, University of Gda\'nsk, 80-308 Gda\'nsk, Poland}

\author{T. S. Mahesh}
\affiliation{Department of Physics, Indian Institute of Science Education and Research, Pune 411008, India}

\maketitle

\begin{abstract}
  Quantum entanglement is a form of correlation between quantum particles that cannot be increased via local operations and classical communication.
  It has therefore been proposed that an increment of quantum entanglement between probes that are interacting solely via a mediator implies non-classicality of the mediator.
  Indeed, under certain assumptions regarding the initial state, entanglement gain between the probes indicates quantum coherence in the mediator.
  Going beyond such assumptions, there exist other initial states which lead to entanglement between the probes via only local interactions with a classical mediator. 
  In this process the initial entanglement between any probe and the rest of the system ``flows through'' the classical mediator and gets localised between the probes.
  Here we theoretically characterise the maximal entanglement gain via classical mediator
  and experimentally demonstrate, using liquid-state NMR spectroscopy, the optimal growth of quantum correlations between two nuclear spin qubits interacting through a mediator qubit in a classical state. We additionally monitor, i.e., dephase, the mediator in order to emphasise its classical character. 
  It is important to note here that no new entanglement is being generated, but rather the correlation already present in the system is being redistributed or localised between the two probe qubits.   Our results indicate the necessity of verifying features of the initial state if entanglement gain between the probes is used as a figure of merit for witnessing non-classical mediator.
  Such methods were proposed to have exemplary applications in quantum optomechanics, quantum biology and quantum gravity.
\end{abstract}

\section{Introduction}
Quantum entanglement is widely recognised as a resource ``as real as energy''~\cite{RevModPhys.81.865}.
Yet, limits on establishing entanglement between remote particles were systematically studied only recently and with surprising results.
Protocols in which the distant particles get entangled by exchanging ancillary particles can establish remote entanglement without ever communicating it, i.e., no entanglement with the ancillaries~\cite{cubitt,exp0,exp1,exp2,exp3}.
It is now understood that entanglement gain in these schemes is not bounded by the communicated entanglement, but rather by communicated quantum discord~\cite{bounds1,bounds2,Streltsov14,Streltsov15,Zuppardo16,Streltsov-review}, 
a form of quantum correlation that persists in many disentangled states~\cite{henderson2001classical,discord,discord-review}.

In another route to producing remote entanglement, the exchange of quantum particles is replaced by continuous interactions of distant systems (probes) with a third object, a mediator.
In this scenario the theory predicts that not only the probes can get entangled without ever entangling the mediator~\cite{cubitt},
but also that they can even get entangled in the absence of any quantum discord between the probes and the mediator~\cite{krisnanda2017}.
This lack of discord is a strong notion of classicality which means that the mediator can be measured at any time without disturbing the whole multipartite system.
It is the same notion as ``classical communication'' in the framework of local operations and classical communication at the core of entanglement theory~\cite{locc}, but generalised to continuous in time interactions.
In practical terms, the probes get entangled even if the mediator is continuously monitored or dephased. 

It is an observation of this effect, for a discrete number of measurements on the mediator, that is reported here together with theoretical characterisation of maximal amount of entanglement that can be established in this way. Moreover, in our experiments the monitoring measurement is the same at all times, which reinforces classicality of the mediator being at all times in one of two distinguishable states (correlated to the probes). Additionally to observing exotic effect of multipartite entanglement our results have practical implications.
The scenario where two objects are coupled via a mediator is common in science. For example, entanglement gain between the probes has been proposed as a witness of quantum mediator in various scenarios~\cite{TanKris}, 
such as spin chains in solid state magnetic compounds~\cite{Kon}, mechanical membranes inside optical cavities~\cite{krisnanda2017}, detection of initial quantum correlations in open-system dynamics~\cite{krisnanda2017},
quantum gravity \cite{krisnanda2017,Bose2017,MV2017} or quantum properties of parts of living organisms coupled to a multi-mode cavity field~\cite{krisnanda_bio}.
Present results emphasise that these methods must verify features of the initial state in order to validate their implications, i.e., non-classicality of the mediating system.

\section{Results}
\subsection{Theoretical example}

The simplest example of the discussed phenomenon involves three qubits (spin-$\frac{1}{2}$ systems) in the following initial state~\cite{krisnanda2017}:
\begin{equation}
\rho_0 = \frac{1}{2} | \psi_+ \rangle \langle \psi_+ | \otimes | + \rangle \langle + | \,  +  \frac{1}{2} | \phi_+ \rangle \langle \phi_+ | \otimes | - \rangle \langle - |,
\label{rho_in}
\end{equation}
where the first two qubits are the probes $A$ and $B$, and the third qubit is the mediator $M$.
Kets $| \pm \rangle$ denote the eigenstates of $\sigma_x^M$ Pauli matrix,
whereas $| \psi^+ \rangle = \frac{1}{\sqrt{2}}(| 01 \rangle + | 10 \rangle)$ and 
$| \phi^+ \rangle = \frac{1}{\sqrt{2}}(| 00 \rangle + | 11 \rangle)$ are the two Bell states.
Since one could dephase the mediator in the $\sigma_x^M$ basis without perturbing the total state, the mediator is said to be in a classical form.
Note also that initially the probes are not entangled as their state is an even mixture of two Bell states~\cite{VP1998}, but the whole tripartite state is initially entangled across partitions $A:MB$ and $AM:B$. This system evolves under Hamiltonian ($\hbar = 1$ throughout the paper):
\begin{equation}
\mathcal{H} = \omega(\sigma_x^A + \sigma_x^B) \otimes \sigma_x^M,
\label{EQ_H_EX}
\end{equation}
where each probe individually interacts with the mediator via a coupling constant $\omega $, but not directly with each other.
It is easy to see that the state of the mediator is stationary and hence it remains classical at all times.
Furthermore, at all times, it is one and the same measurement i.e., dephasing along $\sigma_x^M$ basis, that keeps the total state invariant.
Yet entanglement between the probes increases and they become even maximally entangled~\cite{krisnanda2017}.

At first sight this example might be surprising as it seems that entanglement between the probes is increased by local interactions with a classical mediator, in contradiction to the very definition of entanglement~\cite{RevModPhys.81.865,RMP.91.025001}.
We stress that there is no contradiction as already in the initial state each individual probe is entangled with the rest of the system. One can show that the corresponding entanglement, as quantified by negativity \cite{neg}, is given by $E_{A:MB} = E_{AM:B} = 1/2$, which is the amount of entanglement in maximally entangled state of two qubits.
The subsequent evolution localises this entanglement to the probes. One could also think about this process as continuous in time entanglement distillation, where the local interactions ``read'' the state of the mediator and suitably adapt the dynamics to arrive at entanglement between the probes.
It is the essence of our demonstration that entanglement localisation can be done via the classical mediator even if it is measured or dephased.

\subsection{Optimality}

We show in the Methods section that
a resource behind entanglement localisation via a classical mediator is the amount of initial correlations with the mediator.
The amount of entanglement that can be localised is bounded as follows:
\begin{equation}
\mathcal{E}_{A:B}(t) - \mathcal{E}_{A:B}(0) \le I_{AB:M}(0),
\label{EQ_I_BOUND}
\end{equation}
where $\mathcal{E}_{A:B}$ denotes the relative entropy of entanglement~\cite{REE} and $I_{AB:M}$ is the mutual information~\cite{NC-book}. Let us put this result in perspective and repeat, this time with quantitative statements, what has been known about the resources behind distribution of various correlations.
Consider first two laboratories (operated by Alice and Bob) that exchange a quantum particle $M$. 
Initially, Alice has systems $A$ and $M$, then $M$ is communicated to Bob, who finally stores systems $M$ and $B$ in his laboratory. 
How much can information between the laboratories increase in this process?
As expected, the bound is given by the communicated information and is characterised by the inequality
$I_{A:MB} - I_{AM:B} \le I_{AB:M}$~\cite{bounds2}.
Surprisingly, a similar inequality does not extend to quantum entanglement.
Instead, we have $\mathcal{E}_{A:MB} - \mathcal{E}_{AM:B} \le D_{AB|M}$, i.e., entanglement gain between the laboratories is bounded by the communicated quantum discord~\cite{bounds1,bounds2,Streltsov14,Streltsov15,Zuppardo16,Streltsov-review}. This allows for entanglement distribution via exchange of unentangled particles~\cite{cubitt} since discord can be nonzero in separable states.
Now, if the communication is continuous in time (not exchange of a particle) and we study entanglement between the probes only (not entanglement between the laboratories) Eq.~\ref{EQ_I_BOUND} shows that the initial information with  mediator is the relevant bound.
This empowers entanglement localisation via classical mediator because there exist correlated states without discord.

Summing up, the necessary conditions for an initial state to give rise to entanglement localisation via classical mediator include: (i) correlations to mediator, i.e. $I_{AB:M} > 0$, and (ii) entanglement between a probe and the rest of the system, e.g., $E_{A:MB} > 0$.
The latter is formally recognised in Eq.~(\ref{EQ_REVEALING}) of the Methods section, which proves $E_{A:B}(t) \le E_{A:MB}(0)$.
This holds for any convex entanglement monotone.
It is straightforward to verify that 

the theoretical example presented above achieves both the upper bound in Eq.~\ref{EQ_I_BOUND} and the bound given by the initial entanglement with a probe particle.

\subsection{NMR setup}

We use liquid-state NMR spectroscopy of $^{13}$C, $^1$H and $^{19}$F in dibromofluoromethane dissolved in acetone with linear topology, H - C - F (see Fig.~\ref{FIG_INITIAL}(a)).
Nuclei of hydrogen and fluorine are identified as probes $A$ and $B$, respectively, whereas carbon nucleus is naturally the mediator $M$. Experiments were performed in $500$ MHz Bruker NMR spectrometer at room temperature.
The sample consists of identical and fairly isolated dibromofluoromethane molecules and all the dynamics of the three-qubit system is completed before any significant environmental influences~\cite{NMR1,NMR2,NMR3}.
The longitudinal and transverse relaxation time constants are longer than $2$ s and $0.2$ s, respectively.
The internal Hamiltonian of the three spin system in a frame rotating about the Zeeman field with individual Larmor frequencies reads:
\begin{equation}
\begin{split}
H_0 = \frac{\pi}{2} \Big ( J_{AM} \, \sigma_z^A \otimes \sigma_z^M +  J_{BM} \, \sigma_z^B \otimes \sigma_z^M +\\
 J_{AB} \, \sigma_z^A \otimes \sigma_z^B \Big ),
\end{split}
\label{Mol_ham}
\end{equation} 
with $J_{AM}=224.5$ Hz, $J_{BM}=-310.9$ Hz, and $J_{AB}=49.7 ~\mathrm{Hz}$ being the corresponding coupling constants between the nuclei. The qubits in the molecular system have internal dynamics that directly couples the probes $A$ and $B$. The effects of this coupling must be canceled if we are to claim generation of entanglement between the probes via classical mediator only. Thus, during experiments, to evolve the system under the interaction Hamiltonian $\mathcal{H}$ in Eq.~\ref{EQ_H_EX} we switch off the internal interaction between spins $A$ and $B$ by applying suitable refocusing pulses as will be explained later.

In general, the quantum state of our three-qubit NMR system is of the form  \noindent $(1 - \epsilon)\frac{1}{8} \mathbb{1} + \epsilon \, \rho$, 
where $\frac{1}{8} \mathbb{1}$ describes the background population, $\rho$ is the so-called deviation density matrix and $\epsilon$ is the purity factor, which is in the order of $10^{-5}$. Nevertheless the NMR experiments are sensitive to the deviation density matrix and  from now on whenever we refer to the ``state'' of the system we mean the pseudopure state characterized by the deviation density matrix. Despite the low purity of the NMR qubits, it has been shown that such  systems can still efficiently implement certain quantum information tasks, for which no known efficient classical algorithms exist \cite{morimae2017power} 


Starting from a three qubit thermal equilibrium state of longitudinal magnetization at room temperature, we prepare the state corresponding to $| 00 \rangle \langle 00| \otimes \mathbb{1}/2$, written in the order $ABM$, using a similar pulse sequence as given in \cite{SciDir,Mahesh2013}. 
The initial state $\rho_0$, Eq.~\ref{rho_in}, is then obtained by the succession of gates given in Fig.~\ref{FIG_INITIAL}(b).
All gates are implemented using radio frequency pulses resonant with the nuclei. The open CNOT gate is realised with the help of Krotov optimisation technique~\cite{krotov2008quantum} with fidelity exceeding 0.99 using push-pull optimization of quantum controls \cite{batra2020push}. The fidelity of the produced initial state to the ideal one is more than $97\%$.



Fig.~\ref{FIG_INITIAL}(c) shows the pulses used to realize the interaction Hamiltonian, $\mathcal{H}$ in Eq.~\ref{EQ_H_EX}. In our experiments we have set the strength of the coupling constant $\omega = 1$ rad/s. The solid bars and empty bars represent $\pi/2$ and $\pi$ pulses, respectively. The first half of the pulse sequence evolves the system under the coupling between $M$ and $A$. 
Since we have $\sigma_z \otimes \sigma_z$ coupling in our system to start with (see Eq.~\ref{Mol_ham}), 
the $(\pi/2)_y$ pulses transform the $ z $-basis to $ x $-basis which then evolves under $\sigma_z \otimes \sigma_z$ coupling. 
Since the Hamiltonian in Eq.~\ref{EQ_H_EX} is a sum of two commuting terms, we first evolve the entire system under $\sigma_x^A \otimes \sigma_x^M$
followed by the evolution under $\sigma_x^B \otimes \sigma_x^M$ for the same amount of time, 
i.e, the physical time rescaled by the coupling strengths $J_{AM}$ and $J_{BM}$.
As all the three qubits are coupled to each other, we decouple $B$ during the first half of the evolution by refocusing it using a $\pi$ pulse, as shown in Fig.~\ref{FIG_INITIAL}(c). The net effect is that the system only evolves under $AM$ coupling, whereas $B$ remains unaltered. The same is repeated in the second half of the evolution with $A$ being refocused and the system evolving under $BM$ coupling. 
We repeat the experiment with the same initial conditions and different duration of dynamics in order to illustrate how entanglement accumulates between the probes.
The probes in principle gain maximal entanglement at $\omega t = \pi/8$.
Finally, we obtain the deviation density matrices via full state tomography using eleven detection experiments~\cite{NMR-tomo}.

To make the claim that the mediator $M$ is classical even stronger, we introduce another set of experiments in which we dephase (measure) the mediator qubit in between and at the end of the evolution. 
The pulse sequence implementing the dephasing of $M$ 
is depicted in the orange box of Fig.~\ref{FIG_INITIAL}(c).
In contrast to the previous set of experiments, just after the realization of $\sigma_x^A \otimes \sigma_x^M$ the mediator qubit $M$ is dephased in $x$-basis. 
The selective dephasing of the mediator is achieved by a pair of opposite pulsed-field-gradients (PFGs)  separated by a $\pi$-pulse on the mediator.  The PFGs cancel each other for the probes $A$ and $B$, whereas they add-up for $M$.  A $\pi$-pulse on $M$ is applied to undo the spin-flip caused by the previous $\pi$-pulse.  Finally, measurement along $x$-basis is realized by simply rotating the basis using $(\pi/2)_y$ and $(\pi/2)_{-y}$ pulses as shown.

\begin{figure}[!t]
	\centering
	\includegraphics[trim=4.5cm 0cm 6cm 0cm,clip=true,width=8.5cm]{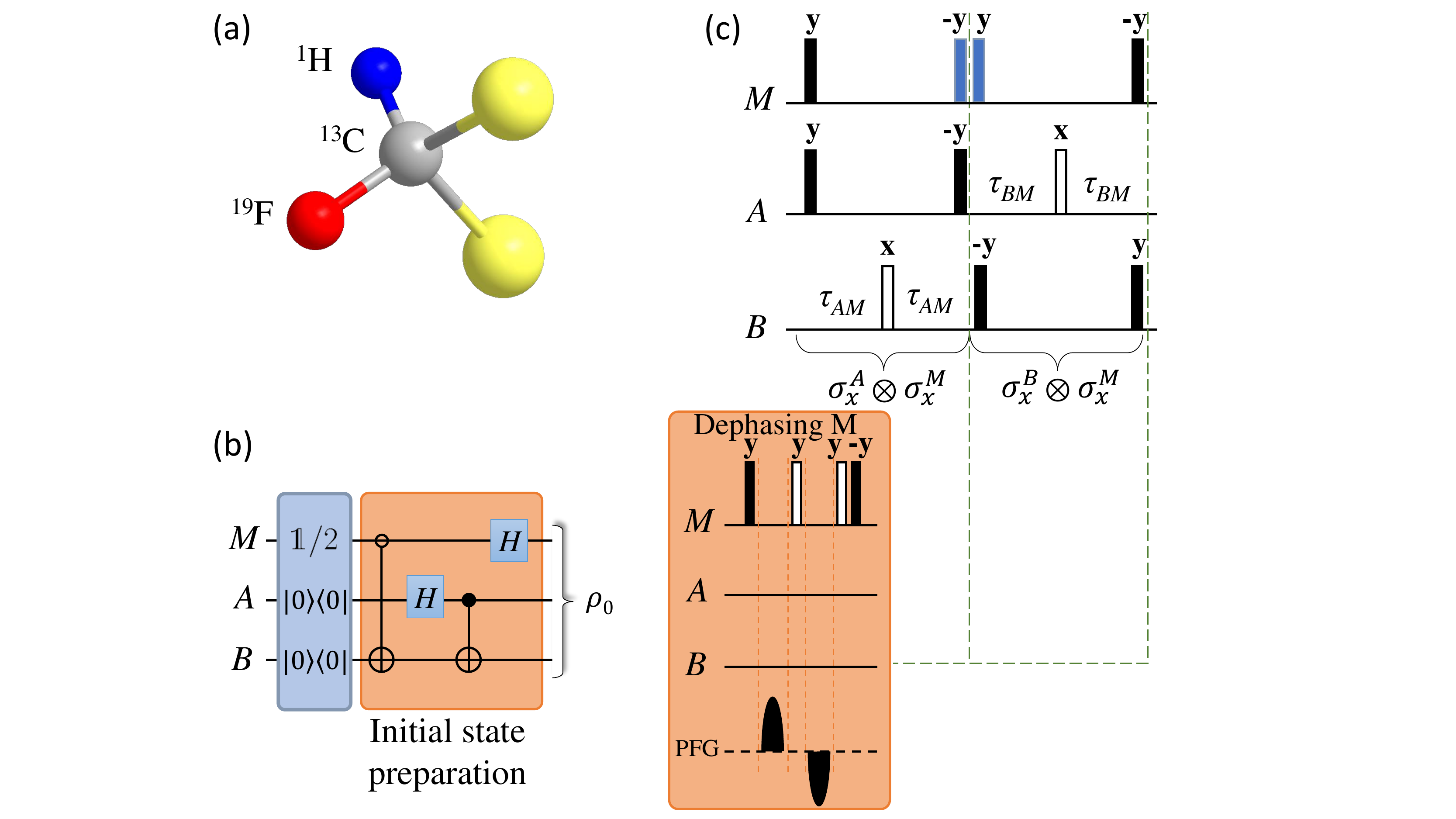}
	\caption{(a) Molecular structure of dibromofluoromethane. We identify $^1$H and $^{19}$F nuclei as probe qubits and $^{13}$C as the mediator qubit. (b) Preparing the initial state as in Eq.~\ref{rho_in} using CNOT and Hadamard gates as shown. 
		(c) Pulse sequences used to evolve the system under the coupling Hamiltonian of Eq.~\ref{EQ_H_EX}. The solid and empty bars represent $\pi/2$ and $\pi$ pulses with phases shown above them.  The blue pulses cancel each other for the no-dephasing case.  Dephasing of $M$ is realized by introducing pulses shown in the orange box in the positions marked by the dashed lines.  Here PFG represents the pulsed-field gradient along $\pm z$ axis.  The delays $\tau_{AM} = 1/(4J_{AM})$ and $\tau_{BM} = 1/(4J_{BM})$.} 
	\label{FIG_INITIAL}
\end{figure}


\section{Discussion}

\begin{figure}
	\centering
	\includegraphics[trim=0.5cm 6.7cm 1.7cm 7.5cm,clip=true,width=9cm]{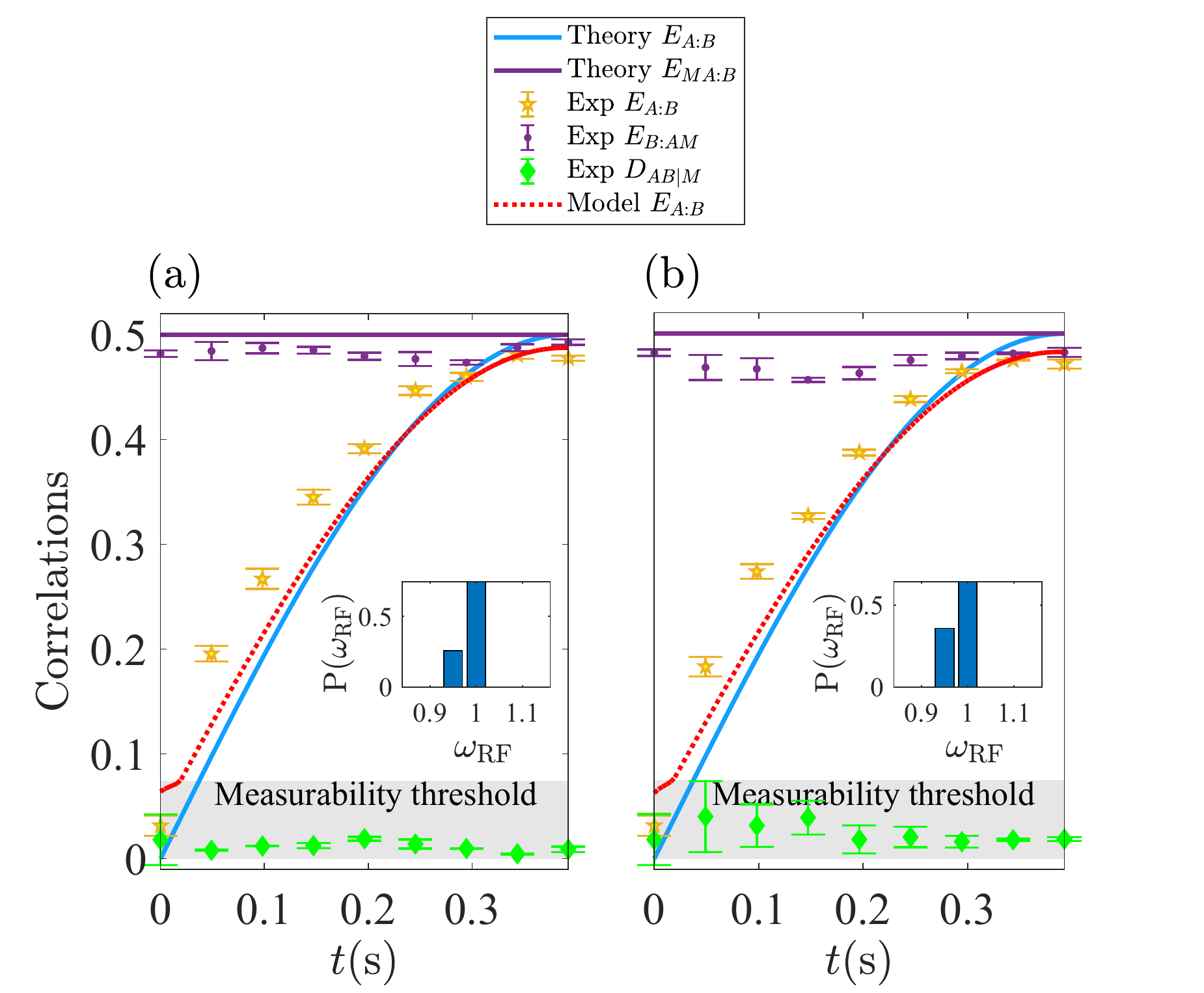}
	\caption{Theoretical (lines) and experimental (markers) correlations: (a) without dephasing and (b) with dephasing the mediator qubit.
		Solid lines show noiseless theoretical predictions: blue for $E_{A:B}$ and purple for $E_{AM:B}$.
		The corresponding experimental data is marked with yellow and purple markers, respectively.
		Green markers show measured discord $D_{AB|M}$, all within experimentally established region of vanishing discord (grey, see main text). The error bars represent the random error in the experiments, obtained from the signal to noise ratio for each experiment. The dotted red lines represent $E_{A:B}$ within a model including inhomogeneity of radio frequency pulses with the distributions $\mathrm{P}(\omega_{\rm RF})$ as shown in the insets.}
	\label{FIG_RESULTS}
\end{figure}

From the experimentally measured three-qubit deviation density matrices we compute various quantum correlations such as discord between the two probes and mediator, $D_{AB|M}$, quantum entanglement between the probes, as measured by the negativity $E_{A:B}$, as well as the negativity $E_{B:AM}$. The quantum discord is calculated following the definition of Ollivier and Zurek~\cite{discord}. 
Recall that discord is not a symmetric quantity and $D_{AB|M}$ denotes discord as measured on the mediator.
It should also be stressed that due to small admixture of the deviation density matrix, the ensemble averaged NMR signals mask genuine entanglement~\cite{Braunstein1999}.
From this perspective one can think of our experiment as NMR simulation of entanglement localisation via classical mediator.

The measured discord and entanglement  are presented in Fig.~\ref{FIG_RESULTS} for datasets without and with dephasing the mediator.
The gray-shaded region represents measurability threshold of discord owing to experimental errors.
This threshold is obtained from measurements of discord for experimental thermal equilibrium state. 
Ideally this state has vanishing discord but experimental imperfections in state tomography give rise to residual values.
The amount of discord $D_{AB|M}$ calculated for evolved deviation density matrices (green data points) all lie well within this experimental precision limit of discord.
We thus conclude that the mediator was classical at all times during the evolution. Yet, negativity of quantum entanglement between the probes $E_{A:B}$ consistently grows
as shown by experimentally estimated values depicted with yellow markers.
The error-bars represent random errors and do not exactly match the noiseless prediction based on Hamiltonian in Eq.~\ref{EQ_H_EX}, presented with a solid blue line.
This is the effect of systematic errors which include the spatial RF inhomogeneity (RFI) over the sample volume, fluctuations in the resonance frequencies due to variations in ambient temperature, as well as finite precision involved in quantum state tomography.
We have simulated the effect of RFI in the evolution as well as tomography pulses, see the dotted red curves,
and indeed observe better agreement with the experimental data.
It is interesting to note that the entanglement between each probe with the rest of the system, $E_{B:AM}$ or $E_{A:BM}$, remains invariant throughout the evolution.
We demonstrated that an increment of quantum entanglement between two probes coupled via a mediator in general does not signify a non-classical mediator.

We note that the process demonstrated here is different from, e.g., entangling two spins via dipole-dipole interaction.
The usual Hamiltonian of the latter directly couples magnetic moments of the two particles and hence it is not surprising that entanglement grows.
In contradistinction, we study tripartite system with an explicit mediator.
Even if the dipole-dipole interaction is rewritten in the form where the mediating virtual photons are clearly distinguished, they typically get entangled with the particles and hence the mediator is not classical.

Perhaps the most interesting application for revealing non-classicality of mediators is witnessing quantum gravity through entanglement between nearby masses~\cite{krisnanda2017,Bose2017,MV2017}. Assuming that the whole mass-gravity-mass system can be described by the standard tripartite Hilbert space formalism, entanglement $E_{A:MB}$ or $E_{AM:B}$ cannot grow via classical mediator, i.e., when the discord $D_{AB|M} = 0$ at all times~\cite{krisnanda2017}. However, the masses alone can become entangled, i.e., $E_{A:B}$ may grow even via classical mediator --- analogously to the phenomenon demonstrated in the present work.
As shown here, this requires initial entanglement $E_{A:MB}$ or $E_{AM:B}$ and correlations to the mediator.
Since in practice one would only measure the masses and not the field, it is highly desirable to provide the bound on possible entanglement gain via classical mediator solely in terms of quantities measured on the masses.
Ref.~\cite{krisnanda2017} proves that the relevant bound is given by the sum of initial entropies of the masses.
Estimations with concrete experimental arrangements show that in order to observe gravitational entanglement the masses need to be cooled down near the ground state of their traps~\cite{Bose2017, Mann2018,arXiv:1812.09776, arXiv1906.08808}.
In such a case the masses are close to a pure state and hence they are initially almost uncorrelated from the rest of the world, i.e., their entropies are small, but not zero. It is this latter bound on entanglement that has to be experimentally violated in order to witness quantum gravity within this framework.

\section{Methods}

We prove here Eq.~\ref{EQ_I_BOUND} of the main text. 
Let us begin with a useful lemma.

\begin{lemma}
	For a tripartite system with classically correlated mediator, i.e, in a state $\rho = \sum_m p_m \: \: \rho_{AB|m} \otimes | m \rangle \langle m |$ with orthonormal basis $\{ | m \rangle \}$, the relative entropy of entanglement follows the bound
	\begin{equation}
	\mathcal{E}_{A:BM} - \mathcal{E}_{A:B} \le I_{AB:M},
	\end{equation}
	where $I_{AB:M}$ is the mutual information between the mediator and remaining systems.
\end{lemma}

\begin{proof}
	From the definition of relative entropy of entanglement, we have $\mathcal{E}_{A:BM} = - \mathrm{tr}( \rho \log \sigma) - S_{ABM}$, where $\sigma$ is the closest separable state to $\rho$ and $S_{ABM}$ stands for von Neumann entropy of state $\rho$~\cite{REE}.
	The flags condition~\cite{flags} applied to $\rho$ gives $\mathcal{E}_{A:BM} = \sum_m p_m \: \mathcal{E}_{A:B}(\rho_{AB|m})$.
	Furthermore, the same reference shows that the corresponding closest separable states satisfy
	$\sigma = \sum_m p_m \sigma_{AB|m} \otimes | m \rangle \langle m |$, where $\sigma_{AB|m}$ is the closest separable state to $\rho_{AB|m}$.
	Using expressions for $\sigma$ and $\rho$, we find 
	\begin{eqnarray}
		\mathcal{E}_{A:BM} & = & S_M - S_{ABM} - \nonumber \\
		&&\sum_m p_m \mathrm{tr}(\rho_{AB|m} \log \sigma_{AB|m}) \\
		\mathcal{E}_{A:B} & = & - \sum_m p_m \mathrm{tr}(\rho_{AB|m} \log \sigma_{AB}) - S_{AB},
	\end{eqnarray}
	where $\sigma_{AB}$ is the closest separable state to the marginal $\rho_{AB} = \sum_m p_m \rho_{AB|m}$.
	Using the definition of mutual information, $I_{AB:M} = S_M + S_{AB} - S_{ABM}$, we have
	\begin{eqnarray}
		\mathcal{E}_{A:BM} - \mathcal{E}_{A:B} & = & I_{AB:M} + \nonumber \\
		&&\sum_m p_m [ -\mathrm{tr}(\rho_{AB|m} \log \sigma_{AB|m}) ] \nonumber \\
		& - & \sum_m p_m [ - \mathrm{tr}(\rho_{AB|m} \log \sigma_{AB})].
	\end{eqnarray}
	The lemma follows by noting that the difference between the last two sums is non-positive because each $\sigma_{AB|m}$ minimises the relative entropy of entanglement of the corresponding $\rho_{AB|m}$.
\end{proof}

Eq.~\ref{EQ_I_BOUND} is the result of the following theorem.
\begin{theorem}
	In a tripartite system with classically correlated mediator at all times (each subsystem can be open to its local environment) we have
	\begin{equation}
	\mathcal{E}_{A:B}(t) - \mathcal{E}_{A:B}(0) \le I_{AB:M}(0),
	\end{equation}
	with notation as in the lemma.
\end{theorem}

\begin{proof}
	Consider the following chain of inequalities:
	\begin{eqnarray}
	\mathcal{E}_{A:B}(t) - \mathcal{E}_{A:B}(0) & \le & \mathcal{E}_{A:BM}(t) - \mathcal{E}_{A:B}(0) \\
	& \le & \mathcal{E}_{A:BM}(0) - \mathcal{E}_{A:B}(0) \label{EQ_REVEALING} \\
	& \le & I_{AB:M}(0).
	\end{eqnarray}
	In the first line we used monotonicity of entanglement under tracing out the mediator $M$, 
	i.e, $\mathcal{E}_{A:B}(t) \le \mathcal{E}_{A:BM}(t)$.
	Inequality in Eq.~\ref{EQ_REVEALING} is proven in Ref.~\cite{krisnanda2017} and states that entanglement in partition $A:BM$ (or $B:AM$) cannot grow via classical mediator, 
	i.e, $\mathcal{E}_{A:BM}(t) \le \mathcal{E}_{A:BM}(0)$.
	Lemma 1 confirms the last line.
\end{proof}

\section{Acknowledgments} 
TP acknowledges discussions with Hermann Kampermann. We thank the organisers of QIPA 2018 where this collaboration was initiated. This work is supported by the Singapore Ministry of Education Academic Research Fund Tier 2, Project No. MOE2015-T2-2-034 and Polish National Agency for Academic Exchange NAWA Project No. PPN/PPO/2018/1/00007/U/00001. TK thanks Timothy C. H. Liew for hospitality at Nanyang Technological University. SP and PR acknowledge discussions with V. R. Krithika. TSM acknowledges the support from the Department of Science and Technology, India (Grant Number DST/SJF/PSA-03/2012-13) and the Council of Scientific and Industrial Research, India (Grant Number CSIR-03(1345)/16/EMR-II).

\bibliographystyle{plainnat}
\bibliography{localisation_revised}
\end{document}